\documentclass[lettersize,journal]{IEEEtran}
\usepackage{amsmath, amsthm, amsfonts,amssymb}
\usepackage{algorithmic}
\usepackage{algorithm}
\usepackage{array}
\usepackage[caption=false,font=normalsize,labelfont=sf,textfont=sf]{subfig}
\usepackage{textcomp}
\usepackage{stfloats}
\usepackage{url}
\usepackage{verbatim}
\usepackage{graphicx}
\usepackage{cite}
\usepackage{scalerel}[2016-12-29]

\usepackage{stackengine}

\newtheorem{lemma}{Lemma}
\theoremstyle{remark}
\newtheorem{remark}{Remark}

\usepackage{enumitem}   
\usepackage{bigints}
\usepackage{xcolor}

\usepackage[numbers,sort&compress]{natbib}

\begin{document}

\title{\LARGE{Analytical Characterization of the Operational \\Diversity Order in Fading Channels}}

\author{Santiago Fernández, J. Alfonso Bail\'on-Mart\'inez, Juan E. Galeote-Cazorla, and F. Javier López-Martínez,~\IEEEmembership{Senior Member,~IEEE,}
        % <-this % stops a space
\thanks{Manuscript received July 17, 2024; revised August 12, 2024. The review of this paper was coordinated by Prof. Petros Bithas. This work was funded in part by Junta de Andaluc\'ia through grant EMERGIA20-00297, in part by MCIU/AEI/10.13039/501100011033 through grant PID2020-118139RB-I00, in part by the predoctoral grant FPU22/03392, and in part by \'Icaro program within the Advanced Telecommunication Technologies Excellence Unit at University of Granada.}% <-this % stops a space
\thanks{The authors are with the Dept. Signal Theory, Networking and Communications, Research Centre for Information and Communication Technologies (CITIC-UGR), University of Granada, 18071, Granada, Spain. Contact e-mail: $\rm fjlm@ugr.es$.}
\thanks{{\color{black} \textit{Reproducible research:} Code available at https://github.com/s-ff/ODO.}}
\thanks{Digital Object Identifier 10.1109/XXX.2021.XXXXXXX}}
% The paper headers
\markboth{IEEE Communications Letters,~Vol.~XX, No.~X, November~2024}%
{Fernández \MakeLowercase{\textit{et al.}}: Analytical Characterization of the Operational Diversity Order in Fading Channels}

\maketitle

\begin{abstract}
We introduce and characterize the operational diversity order (ODO) in fading channels, as a proxy to the classical notion of diversity order at any arbitrary operational signal-to-noise ratio (SNR). Thanks to this definition, relevant insights are brought up in a number of cases: (\textit{i}) We quantify that in \textcolor{black}{dominant} line-of-sight scenarios an increased diversity order is attainable compared to that achieved asymptotically, \textcolor{black}{even in the single-antenna case}; (\textit{ii}) this effect is attenuated, but still visible, in the presence of an additional dominant specular component; (\textit{iii}) the decay slope in Rayleigh product channels increases very slowly, never fully achieving unitary slope for a finite SNR.
\end{abstract}

\begin{IEEEkeywords}
Asymptotic analysis, communication theory, diversity order, fading channels, wireless communications.
\end{IEEEkeywords}

\section{Introduction}
\IEEEPARstart{S}{upported} by continuous advances in microelectronics and computing capabilities, the surge of new generations of wireless technologies enables use cases that seemed futuristic only a couple generations ago \cite{Dang2020}. However, before such new technologies become a commercial reality, their performances need to be thoroughly evaluated and assessed, and insights related to their performance scaling laws and operational limits must be well understood. Taking a deep look into communication theory fundamentals, it is undeniable that asymptotic analysis has been a remarkably useful tool for decades to evaluate system performances \cite{Ventura1997}.

The milestone work \cite{Wang2003} set the basis for the asymptotic performance analysis of wireless communication systems. Under reasonably mild conditions related to the smoothness of the probability density function (PDF) of the signal to noise ratio (SNR), error probability measures can be expressed in the form of $P_{\rm op}\approx \alpha\left({\gamma_{\rm th}}/{\overline\gamma}\right)^{b}$ for sufficiently large average SNR $\overline\gamma$, with $\gamma_{\rm th}$ being the threshold SNR value required for a given performance. The notion of \textit{coding gain} or power offset (captured by $\alpha$) and \textit{diversity order} (DO, captured by $b$) have become ubiquitous in the wireless literature, as a way to characterize performance scaling laws: \textit{how much performance increase can we have by increasing the average SNR a certain amount?} Still today, Wang and Giannakis' power law approximation is used to analyze rather complex architectures \cite{New2024} in terms of their asymptotic outage probability\footnote{Asymptotic expressions for the OP and the uncoded symbol error probability have the same functional form, up to some scale factor \cite{Wang2003}. For simplicity, we will focus in the OP measure as a key benchmark for error performances agnostic to the specific modulation scheme.} (OP).

Now, with great power comes great responsibility, and asymptotic analysis tools need to be used with caution \cite{Dohler2011} to provide valuable guidelines instead of misleading conclusions. First, performances predicted by asymptotic analysis may not be relevant in practice \textcolor{black}{due to} a number of factors: (\textit{i}) low-SNR operating points for users in cellular systems \cite{Dohler2011}, (\textit{ii}) the asymptotic approximation kicks in only at extremely low OP values \cite{Eggers2019} that are not operational even in ultra-reliable regimes. This latter effect becomes relevant especially in line-of-sight (LoS) scenarios, or in the presence of reduced scattering \cite{Ramirez2020}, and can provide false intuitions on the error performance decay at operational SNRs. Second, even though the conditions established in \cite{Wang2003} were mild, they are not met in quite some relevant practical cases. For instance, the DO in keyhole multiple-input multiple-output (MIMO) channels does not admit a power law approximation for equal number of transmit and receive antennas \cite{Sanayei2007}. This includes as special case the Rayleigh product (or cascaded) channel {\cite{Erceg1997}, which is a key building block in relay systems \cite{Hasna2003} and backscatter communications \cite{Griffin2008}. Similarly, the lognormal distribution neither admits a power law approximation under the conditions established in \cite{Wang2003}. For this reason, some alternative asymptotic metrics have been provided to circumvent this issue \cite{Safari2008,Elamassie2020}, and the need for a proxy to DO in finite/operational SNR values has been identified \cite{Narasimhan2006}. 

Aiming to provide further analytical support to this latter work, where the notion of \textit{effective} diversity order was introduced in the context of the diversity-multiplexing trade-off in MIMO systems, we formalize the definition of operational diversity order (ODO) in the context of fading channels, understood as the decay slope for a linear approximation to the log-CDF around the operating SNR. We provide a simple closed-form expression for the ODO valid for any choice of fading distribution, thus not being restricted to the constraints in \cite{Wang2003}. Then, we provide analytical evidences that some fading distributions offer increased diversity orders in \textcolor{black}{dominant} LoS conditions, analyzing the relevant cases of Rician and Two-wave with Diffuse Power (TWDP) fading channels \cite{Durgin2002}. \textcolor{black}{The validity of the ODO is exemplified in a multi-antenna context for selection combining (SC) and maximal ratio combining (MRC) schemes. Finally, the} case of the Rayleigh product channel is analyzed, indicating that the unitary DO is never fully achieved for practical (operational) SNRs.

\vspace{-2mm}
\section{System model and definitions}
Let us define the instantaneous SNR at the receiver side as $\gamma = g \overline{\gamma}$, where $\overline{\gamma}$ is a deterministic positive quantity representing the average received SNR, while $g$ is a channel-dependent non-negative normalized random variable capturing the effects of fading\footnote{\textcolor{black}{While $\gamma$ is a scalar, it can include any combination of transmit and receive filters in multi-antenna set-ups, as well as intermediate processing (e.g., relays, or reflecting surfaces) although the complexity of such processing stages will
determine the ultimate statistical nature of $g$, and hence of $\gamma$.}}. Let $R$ denote the transmission rate of information packets; if packet errors occur due to outages caused by fading, the packet error rate (PER) $\epsilon$ is captured by the OP, defined as the probability that the instantaneous SNR $\gamma$ falls below a certain threshold $\gamma_{\rm th}$, i.e.,
\begin{equation}
    \epsilon\triangleq P_{\rm op}= \Pr(\log_2(1+\gamma)<R) = {\Pr}(\gamma<\gamma_{\rm th}),
    \label{eq.epsilon_1}
\end{equation}
\noindent where $\gamma_{\rm th} = 2^R - 1$ is the minimal required SNR to receive the packet sent at a rate $R$. 
Assuming narrowband channel models with block fading, the power at which the packet is received remains constant and equal to $\Omega$, so that $\overline\gamma=\Omega/N_{\rm 0}$, with $N_{\rm 0}$ being the noise power at the receiver. Equivalently, the received signal power can be expressed as $W=g\Omega$, since $W$ and $\gamma$ are scaled versions of each other through $N_{\rm 0}$, i.e., $\gamma=W/N_{\rm 0}$, and hence a minimal required power to decode
the packet correctly at rate $R$ is given as $W_{\rm th}=\gamma_{\rm th}N_{\rm 0}$. With these definitions, it is possible to rewrite the expression \eqref{eq.epsilon_1} as
\begin{equation}
\label{eqop}
     P_{\rm op}=\Pr(W<W_{\rm th})= \int_{g_{\rm min}}^{W_{\rm th}/\Omega} f_g(g) dg = {\rm F}_{g}\left(W_{\rm th}/\Omega\right),
\end{equation}
\noindent where $f_g(\cdot)$ is the PDF of the power channel coefficient $g$, ${\rm F}_{g}\left(\cdot\right)$ denotes its cumulative distribution function (CDF), and $g_{\rm min}$ denotes the lower limit of the support for the distribution of $g$; classically, $g_{\rm min}=0$ for most fading distributions.

As aforementioned, for a wide variety of channel models the behavior of the lower tail can be approximated by a power-law expression \cite{Wang2003,Eggers2019} in the high SNR regime as
\begin{equation}
   {\rm F}_{g}\left(W_{\rm th}/\Omega\right)\approx \alpha\left( \frac{W_{\rm th}}{{\Omega}}\right)^b=\alpha\left( \frac{\gamma_{\rm th}}{{\overline\gamma}}\right)^b, \label{eqcdf}
\end{equation}
\noindent where $\alpha$ and $b$ are parameters that depend on the actual channel model, and \textit{do not depend on $\left\{\overline\gamma,\Omega\right\}$}. For the case of considering the log-CDF $G_g\left({\cdot}\right)= \log_{10}{\rm F}_{g}\left(\cdot\right)$, the power-law approximation becomes linear in log-log scale, as
\begin{equation}
   G_g\left(W_{\rm th}/\Omega\right)\approx \frac{\alpha^{\rm dB}}{10} +\frac{b}{10}\left(W_{\rm th}^{\rm dB}-\Omega^{\rm dB}\right),\label{eqlogcdf}
\end{equation}
\noindent where the superindex ${}^{\rm dB}$ denotes a transformation in the form $10\log_{10}(\cdot)$ over a given magnitude. Expressions \eqref{eqcdf} and \eqref{eqlogcdf} are asymptotic approximations, which become increasingly valid as $\left\{\overline\gamma,\Omega\right\}\rightarrow\infty$, or equivalently, as $\left\{\gamma_{\rm th}/\overline\gamma,W_{\rm th}/\Omega\right\}\rightarrow 0$.

\section{Operational diversity order}
Our goal is to provide an expression similar to \eqref{eqlogcdf}, that is valid for \textit{any operational} received power $\Omega$ or average SNR $\overline\gamma$. To achieve this, we will introduce the notion of \textit{operational} diversity order (ODO) to capture how certain variations around an \textit{operational} SNR impact the error performance in a similar form as the conventional DO does for sufficiently high (often non-operational) SNRs. Let us consider a generic function $\mathcal{G}\left(x;\Gamma\right)$, where $\Gamma$ denotes an arbitrary set of constant parameters. Then, a linear approximation to $\mathcal{G}\left(x;\Gamma\right)$ around the operating point $x=x_{\rm 0}$ has the form
\begin{equation}
	T(x,x_{\rm 0}; \Gamma) = \mathcal{G}'(x_{\rm 0}; \Gamma)\left(x-x_{\rm 0}\right) + \mathcal{G}(x_{\rm 0};\Gamma),\label{eqlinear}
\end{equation}
\noindent where
\begin{equation}
	\mathcal{G}'(x_{\rm 0};\Gamma) = \left. \frac{\partial \mathcal{G}\left(x;\Gamma\right)}{\partial x}\right|_{x=x_{\rm 0}}.\label{eqlinear2}
\end{equation}
Assuming that the generic function $\mathcal{G}\left(x;\Gamma\right)$ corresponds to the log-CDF of $g$, with $x=\Omega^{\rm dB}$ and $x_{\rm 0}=\Omega_{\rm 0}^{\rm dB}$, we formalize a new linear approximation to the OP in the following lemma:
\begin{lemma}[CDF]
\label{lem:LX2}
Let us consider an OP metric as defined in \eqref{eqop}. Then, the OP around the operation point $\Omega_{\rm 0}$ can be approximated as
\begin{equation}
    P_{\rm op}\approx \alpha_{\rm 0}\left(\frac{\Omega_{\rm 0}}{\Omega}\right)^\delta,\label{eq.lemma}
\end{equation}
where
\begin{align}
\alpha_0 & ={\rm F}_{g}\left(W_{\rm th}/\Omega_{\rm 0}\right),\\
\delta & = \frac{W_{\rm th}}{\Omega_{\rm 0}}\frac{ f_{g}\left(W_{\rm th}/\Omega_{\rm 0}\right)}{{\rm F}_{g}\left(W_{\rm th}/\Omega_{\rm 0}\right)}. 
\label{eq.delta} 
\end{align}
\end{lemma}
\begin{proof}
See appendix \ref{app1}. 
\end{proof} 
In Lemma \ref{lem:LX2}, $\alpha_{\rm 0}$ can be interpreted as an operational power offset. This parameter has a similar interpretation as the $\alpha$ parameter in \eqref{eqcdf}, only that evaluating the CDF at a different point, i.e., $\alpha={\rm F}_{g}\left(1\right)$ when $\Omega=W_{\rm th}$. Now, the key parameter $\delta$ can be seen as an ODO, dominating the log-log decay of the OP at a given operational point $\Omega_{\rm 0}$. Unlike the conventional DO, which only depends on $R$ and the distribution parameters of $g$, the ODO also depends on $\Omega_{\rm 0}$; i.e., depending on the operational value of $\Omega$, a different ODO may be attainable. \textcolor{black}{Note that it is also possible to estimate the ODO from data in a model-agnostic fashion (e.g., estimating the empirical PDF and CDF using \texttt{ksdensity} and \texttt{ecdf} in \texttt{MATLAB}, respectively, or directly from (6) using the \texttt{diff} command over the empirical log-CDF)}. Despite its remarkable simplicity, we note that the analytical formulation in Lemma \ref{lem:LX2} is new in the literature to the best of our knowledge.

\begin{remark}[\textit{How to use the ODO for system design?}]
\label{rem:RX1}
Since the linear approximation for the log-CDF in Lemma \ref{lem:LX2} coincides with the exact OP at $\Omega_{\rm 0}$, we have $P_{\rm op}(\Omega_{\rm 0})=\alpha_{\rm 0}=F_g(W_{\rm th}/\Omega_{\rm 0})$. In the proximity of $\Omega_{\rm 0}$, the linear approximation in \eqref{eq.lemma} predicts that a $c$-fold increase ($c>1$) or decrease ($c<1$) in the received power/SNR (i.e., a $c_{\rm dB}=10\log_{10}(c)$ dB increase/decrease) causes the OP to be scaled by a factor of $c^\delta$. Mathematically:
\begin{equation}
\frac{P_{\rm op}(\Omega=\Omega_{\rm 0})}{P_{\rm op}(\Omega=c\cdot\Omega_{\rm 0})}\approx c^{\delta}.
\end{equation}
In other words, the power increase required to achieve a 1 order of magnitude improvement in terms of OP is well-approximated by
\begin{equation}
\frac{P_{\rm op}(\Omega=\Omega_{\rm 0})}{P_{\rm op}(\Omega=c\cdot\Omega_{\rm 0})}=10\Rightarrow c_{\rm dB} \approx \frac{10}{\delta} ({\rm dB}).
\label{power}
\end{equation}
\end{remark}

\section{Application examples.}
\subsection{Effect of \textcolor{black}{a dominant LoS component}.}
Let us first investigate the role of the ODO in a \textcolor{black}{dominant} LoS scenario; for this purpose, we consider the Rician distribution. By virtue of Lemma \ref{lem:LX2}, and using the well-known expressions for the Rician PDF and CDF, the ODO can be expressed in closed-form as
\begin{equation}
    \delta_{\rm Rice} =\tfrac{W_{\rm th}}{\Omega_{\rm 0}} \tfrac{{(1+K)} e^{-K-\tfrac{(1+K)W_{\rm th}}{\Omega_{\rm 0}}}I_0\left( 2\sqrt{\tfrac{ K (1 + K)W_{\rm th}}{\Omega_{\rm 0}}} \right)}{1 - \text{Q}_1\left(\sqrt{2K}; \sqrt{\tfrac{2(1+K)W_{\rm th}}{\Omega_{\rm 0}}}\right)},
    \label{ODO_Rice}
\end{equation}
where $I_{\nu}(\cdot)$ represents the modified Bessel function of the first kind and order $\nu$, whereas $\text{Q}_n(\cdot)$ denotes the $n$-order Marcum-$Q$ function.
In Fig. \ref{fig:1}, the ODO is evaluated as a function of the average operational received power\footnote{In the sequel, for the sake of simplicity we assume a normalized noise power $N_{\rm0 }=1$, so that the values of $\Omega$ and $W_{\rm th}$ can be regarded as relative to $N_{\rm0 }$, and they can be expressed in dB instead of in absolute power units.} $\Omega_{\rm 0}$, assuming a threshold rate of $R=1.7$ bps/Hz, and varying the $K$ factor that quantifies the amount of LoS. After close inspection of the evolution of $\delta_{\rm Rice}$, some remarks are in order: 
\begin{enumerate}[label=(\roman*)]
    \item The ODO converges to the classical DO (unity for the case of Rician fading) as $\Omega_{\rm 0}\rightarrow\infty$.
    \item Note that, for lower operational values of $\Omega$, the ODO can exhibit larger values than unity as a stronger LoS component is considered. This is in coherence with the observation made in \cite{Eggers2019}, which suggested the possibility for the Rician distribution to offer an increased diversity order in practical operational values of OP.
    \item As the LoS component \textcolor{black}{ceases to be dominant compared to the diffuse component}, (i.e., for values of $K\leq1$), the ODO monotonically increases with $\Omega$, never exceeding the asymptotic limit set by the classical DO.
\end{enumerate}
To better understand the role of the ODO, we exemplify in Fig. \ref{fig:OutProb_Rice} the linear approximations to the OP at different operational values of SNR, for the case of $K=15$. Depending on the operational value of the SNR, we see that the decay slope of the linear approximation to the OP differs from that predicted by the conventional DO. For instance, let us focus on the case of an operational value of $\Omega_{\rm 0}=10$ dB. Following the rationale in Remark \ref{rem:RX1}, increasing (equivalently, decreasing) the operational receive power (or equivalently, SNR) by a factor of $2$ (i.e., 3 dB) makes the OP to be reduced (equivalently, increased) by a factor of $2^{\delta_{\rm Rice}}\approx 21$. To improve the OP performance by one order of magnitude, a power increase of $c_{\rm dB}\approx 2.27$ dB is required when operating at $\Omega_{\rm 0}=10$ dB. If instead the operational point is at $\Omega_{\rm 0}=20$ dB, then $c_{\rm dB}\approx 4.25$ dB is required for a one-order of magnitude improvement. Hence, the ODO becomes a useful tool to predict performance excursions around a given operational point, in a more general (and precise) way than the conventional DO counterpart.
\begin{figure}[t]
 \includegraphics[width=0.82\columnwidth]{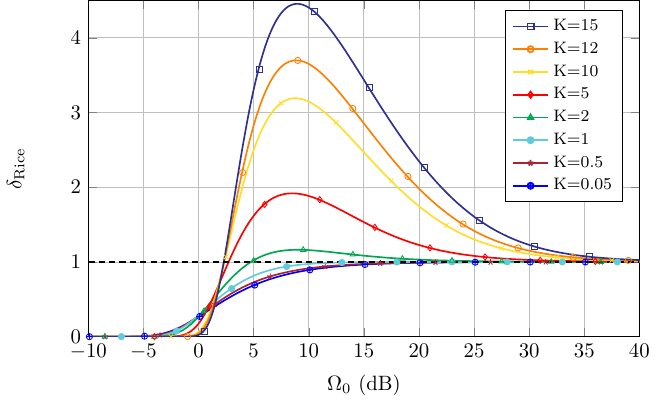}
 \caption{ODO for the Rician fading channel, as a function of $\Omega_{\rm 0}$, for different values of $K$. Solid lines are obtained from \eqref{ODO_Rice}, markers correspond to Monte Carlo (MC) simulations \textcolor{black}{using \texttt{diff} over the empirical log-CDF}. Conventional DO is included \textcolor{black}{as an asymptotic reference (black dashed line) to highlight that the ODO converges to it as $\Omega_{\rm 0}\rightarrow\infty$.}
 }
\label{fig:1}
\end{figure}
\begin{figure}[t]
	\includegraphics[width=0.82\columnwidth]{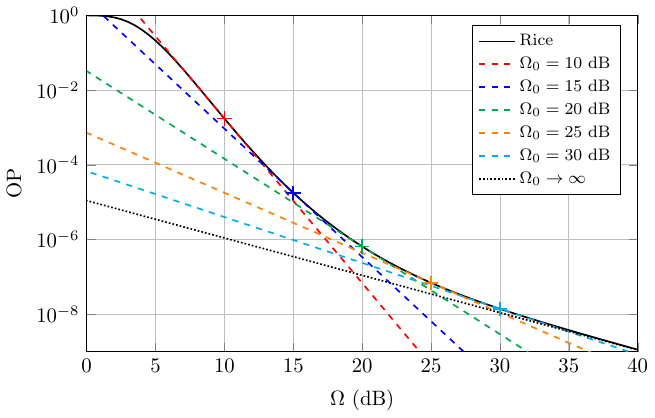}
 \caption{ODO-based linear approximations to the OP for the Rician fading channel at different operational points \textcolor{black}{(each indicated by a colored + mark)}. Parameter values are $R=1.7$ bps/Hz and $K=15$. \textcolor{black}{Solid black line corresponds} to the theoretical expression for the OP.
 .}
\label{fig:OutProb_Rice}
\end{figure}

\subsection{\textcolor{black}{Multiple antennas}}
\textcolor{black}{We now consider the case of multi-antenna reception with i.i.d. branches, considering the reference cases of SC and MRC techniques for exemplary purposes \cite{Brennan1959}. As formally proved in Appendix \ref{app2}, the ODO in an $N$-branch SC scheme is $N$ times the ODO of the single-antenna case. Now, the ODO for the MRC case can be computed from \cite[eq. (27-28)]{Romero2008} as}
\begin{equation}
\textcolor{black}{\delta_{\rm Rice}^{\rm MRC} =\tfrac{W_{\rm th}}{\Omega_{\rm 0}} \tfrac{A_{K,N}\left(W_{\rm th},\Omega_{\rm 0}\right)I_{N-1}\left( 2\sqrt{\tfrac{ K N(1 + K)W_{\rm th}}{\Omega_{\rm 0}}} \right)}{1 - \text{Q}_N\left(\sqrt{2KN}; \sqrt{\tfrac{2(1+K)W_{\rm th}}{\Omega_{\rm 0}}}\right)}},
    \label{ODO_RiceMRC}
\end{equation}
 $A_{K,N}\left(W_{\rm th},\Omega_{\rm 0}\right)=\left(\tfrac{1+K}{\Omega_{\rm 0}} \right)^{\tfrac{N+1}{2}}\left(\tfrac{W_{\rm th}}{NK}\right)^{\tfrac{N-1}{2}}e^{-KN-\tfrac{(1+K)W_{\rm th}}{\Omega_{\rm 0}}}$, which reduces to \eqref{ODO_Rice} for $N=1$. The ODO in the SC and MRC receivers is evaluated in Fig. \ref{fig:new}, assuming $N = 4$. \textcolor{black}{Comparing both cases, we observe: (\textit{i})} both techniques achieve full diversity order (i.e., reaching the value of $N$ as $\Omega_{\rm 0}$ grows); \textcolor{black}{(\textit{ii}) their ODOs have a similar shape, although the maximum value is reached for a smaller $\Omega_{\rm 0}$ in the case of MRC (due to its larger coding gain compared to SC); (\textit{iii})} in both instances, a larger ODO than N is achievable under a dominant LoS component.

\begin{figure}[t]
 \includegraphics[width=0.82\columnwidth]{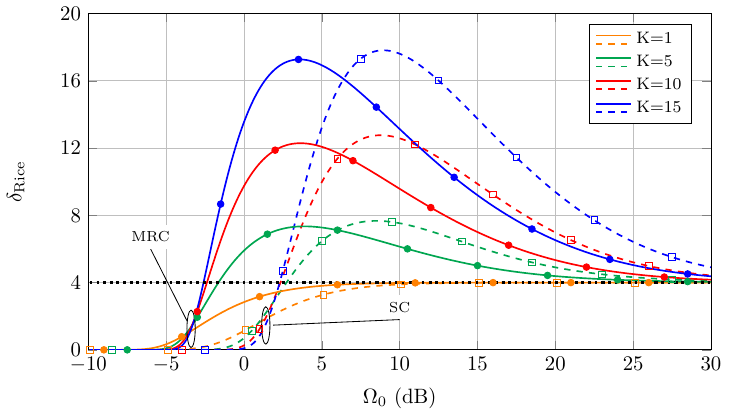}
 \caption{\textcolor{black}{ODO under Rician fading with SC and MRC diversity schemes, for different values of $K$. Solid lines obtained from \eqref{ODO_RiceMRC}, markers indicate MC simulations \textcolor{black}{using \texttt{diff} over the empirical log-CDF}. Conventional DO is included as asymptotic reference (black dashed line) to highlight that the ODO converges to it as $\Omega_{\rm 0}\rightarrow\infty$.}
 }
\label{fig:new}
\end{figure}

\begin{figure}[t]
	\includegraphics[width=0.82\columnwidth]{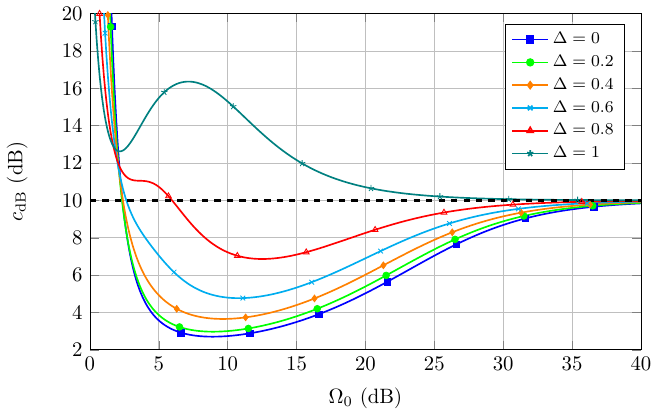}
 	\caption{\textcolor{black}{Power increase (in dB) required to achieve a 1 order of magnitude improvement in OP} for the TWDP fading channel, as a function $\Omega_{\rm 0}$, for different values of $\Delta$. Parameter values are $K=12$ and $R=1.7$ bps/Hz. Solid lines are obtained from \eqref{ODO_TWDP} \textcolor{black}{and \eqref{power}}, markers correspond to MC simulations \textcolor{black}{using \texttt{diff} over the empirical log-CDF}. \textcolor{black}{The value of $c$ predicted by the conventional DO is included as an asymptotic reference (black dashed line) as $\Omega_{\rm 0}\rightarrow\infty$.}}	
\label{fig:NumericalAndTheoretical_Derivatives_TWDP_DeltaVar_K_const}
\end{figure}
\subsection{Effect of a second specular component}
Now, let us move to a more general LoS set-up, by assuming the presence of a second dominant specular component. This is encapsulated by Durgin's TWDP fading model \cite{Durgin2002}, through the parameter $\Delta\in[0,1]$ that captures the amplitude imbalance between these two components. In the limit case of $\Delta=0$, TWDP model reduces to the Rician case, whereas the case with $\Delta=1$ corresponds to a scenario with two LoS components of equal magnitude. Plugging in \eqref{eq.delta} the PDF and CDF expressions in \cite{Rao2015}, we can express
\begin{equation}
\textcolor{black}{    \delta_{\rm TWDP} = \frac{W_{\rm th}\xi_{\rm 0}}{\pi} \frac{\int_{0}^{\pi} e^{-\kappa_{\theta} - {W_{\rm th}\xi_{\rm 0}}}I_0\left( 2 \sqrt{{\kappa_{\theta} \xi_{\rm 0}W_{\rm th}}} \right)}{1 - \tfrac{1}{\pi}  \int_{0}^{\pi} \text{Q}_{1}\left(\sqrt{2 \kappa_{\theta}}; \right.  \left. \sqrt{2\xi_{\rm 0}W_{\rm th}}\right) d\theta}},
    \label{ODO_TWDP}
\end{equation}
\noindent \textcolor{black}{where $\xi_{\rm 0}=\frac{K+1}{\Omega_{\rm 0}}$} and $\kappa_{\theta}=K (1+\Delta \cos \theta)$. To better understand the role of $\Delta$ with respect to the Rician case, the ODO is evaluated in Fig. \ref{fig:NumericalAndTheoretical_Derivatives_TWDP_DeltaVar_K_const} \textcolor{black}{through the power increase metric in \eqref{power},} for different values of $\Delta$. A reference scenario with a dominant LoS propagation is assumed, i.e., $K=12$. We note that as $\Delta$ is increased, the \textcolor{black}{power increase required to improve the OP by one order of magnitude grows. In all instances, such value is heavily dependent on $\Delta$ and, as $\Delta\rightarrow 1$, it can even exceed that predicted by conventional DO ($10$ dB).} In this situation, often associated to a worse-than-Rayleigh condition \cite{Matolak2011} due to the probability of cancellation between the dominant specular components, the ODO for the TWDP case is below unity even for such a large value of $K$.

	\begin{figure}[t]
	\includegraphics[width=0.82\columnwidth]{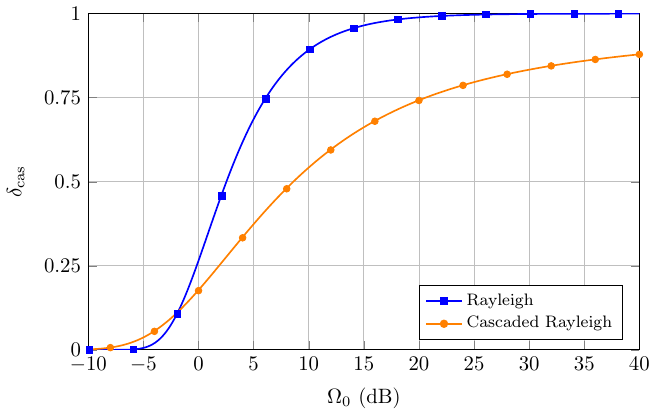}		
	\caption{ODO for the cascaded and single Rayleigh fading channels, as a function of $\Omega_{\rm 0}$. Solid lines are obtained from \eqref{ODO_cas}, markers correspond to MC simulations \textcolor{black}{using \texttt{diff} over the empirical log-CDF}.}
\label{fig:NumericalAndTheoretical_Derivatives_CascadeRay_Ray}
\end{figure}

\subsection{When conventional DO is not applicable}
The general definition of DO in \cite{Wang2003} established a mild condition over the PDF of the SNR, requiring that such a PDF could be approximated by a single polynomial term in the proximity of $0^{+}$. An equivalent condition can be formulated in the Laplace domain, so that the moment generating function (MGF) presents a decay in the form $|s|^{-b}$. However, this condition does not hold in some practical situations, including the case of MIMO keyhole channels with equal number of transmit and receive antennas \cite{Sanayei2007} as a relevant example \textcolor{black}{in double-scattering scenarios \cite{Gesbert2002}}. Let us focus on the baseline case of a single-input single-output scenario, referred to as product (or cascaded) Rayleigh channel \cite{Erceg1997}. In this situation, the MGF presents a decay of the form $\log(s)/s$, and the classical DO fails to capture the asymptotic behavior of the OP in this scenario. Using the well-known expressions for the cascaded Rayleigh PDF and CDF \cite{Eggers2019}, we can obtain a closed-form expression for the ODO as:
\begin{equation}
    \delta_{\rm cas} = \tfrac{W_{\rm th}}{\Omega_{\rm 0}} \tfrac{2 K_0 \left( 2 \sqrt{ \tfrac{W_{\rm th}}{\Omega_{\rm 0}}}\right)}{1 - 2 \sqrt{\tfrac{W_{\rm th}}{\Omega_{\rm 0}}} K_1 \left( 2 \sqrt{ \tfrac{W_{\rm th}}{\Omega_{\rm 0}}}\right)},
    \label{ODO_cas}
\end{equation}
\noindent where $K_{0}(\cdot)$ and $K_1(\cdot)$ correspond to the modified Bessel functions of the second kind and order zero, and one, respectively. The ODO for the cascaded Rayleigh fading channel is represented in Fig. \ref{fig:NumericalAndTheoretical_Derivatives_CascadeRay_Ray}, and compared to the case of a single Rayleigh link, which has a simple form given by
\begin{equation}
\delta_{\rm Ray}=\tfrac{W_{\rm th}}{\Omega_{\rm 0}} \tfrac{e^{-\tfrac{W_{\rm th}}{\Omega_{\rm 0}}}}{1 - e^{-\tfrac{W_{\rm th}}{\Omega_{\rm 0}}}}.
\end{equation} 
\begin{figure}[t]
	\includegraphics[width=0.82\columnwidth]{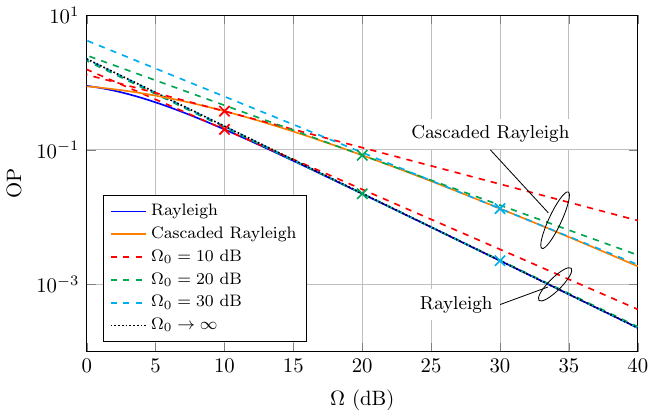}			
 \caption{ODO-based linear approximations to the OP for the single and cascaded Rayleigh fading channels at different operational points \textcolor{black}{(indicated by the colored $\times$ markers)}. Parameter values are $R=1.7$ bps/Hz. Solid lines correspond to the theoretical expressions for the OP.}	\label{fig:OutProb_CascadeRayleigh_Rayleigh}
\end{figure}
While for the latter the ODO quickly converges to unity, the former has a much slower growth and never manages to achieve the same DO as the single Rayleigh link. This is also represented in Fig. \ref{fig:OutProb_CascadeRayleigh_Rayleigh}, where different linear approximations to the log-CDF are represented, using the slope values determined by the ODO. We see that assuming that the cascaded Rayleigh approximately decays with unitary slope may be \textcolor{black}{overly} optimistic \textcolor{black}{(e.g., slope is below 0.5/0.75 for $\Omega_{\rm 0}$ below 8/20 dB, respectively). Hence,} the true decay captured by the ODO must be considered for a proper system design around a given operation point.

\section{Conclusions}
We presented and formally characterized the operational diversity order, as an extension to the classical definition of diversity order but \textcolor{black}{valid in the vicinity of any operational SNR}.  
Thanks to the ODO, the increased diversity order offered by LoS propagation is quantified, and the true diversity order in keyhole channels is analytically computed. The simple and compact definition of the ODO opens up the possibility to directly evaluate this insightful performance metric for any wireless system. \textcolor{black}{Establishing the necessary and sufficient conditions for the ODO to exceed the conventional DO remains an open problem that requires further attention.}

\appendices

\section{Proof of Lemma \ref{lem:LX2}}
\label{app1}

By analogy with \eqref{eqlogcdf}, inspection of \eqref{eqlinear} and \eqref{eq.delta} implies the following equivalence:
\begin{equation}
\label{eqapp1}
    \delta=-10 \left. \frac{\partial \mathcal{G}\left(x;\Gamma\right)}{\partial x}\right|_{x=x_{\rm 0}}=-10 \left. \frac{\partial \log_{10}P_{\rm op}\left(\Omega;\Gamma\right)}{\partial \Omega^{\rm dB}}\right|_{\Omega=\Omega_{\rm 0}},
\end{equation}
where the set of parameters $\Gamma$ includes the distribution parameters for $g$ and the threshold value $W_{\rm th}$ for the sake of shorthand notation. To prove the lemma, the following steps are applied from \eqref{eqapp1}:
\begin{align}
    \delta&=- \left. \frac{\partial \log_{10}P_{\rm op}\left(\Omega;\Gamma\right)}{\partial \log_{10} \Omega}\right|_{\Omega=\Omega_{\rm 0}},\\
    &=- \left. \frac{\partial \log P_{\rm op}\left(\Omega;\Gamma\right)}{\partial \log \Omega}\right|_{\Omega=\Omega_{\rm 0}},\\
    & {\color{black}\stackrel{(a)}{=}} - \Omega \left. \frac{\partial \log P_{\rm op}\left(\Omega;\Gamma\right)}{\partial \Omega}\right|_{\Omega=\Omega_{\rm 0}},\\
    & {\color{black}\stackrel{(b)}{=}} - \frac{\Omega}{\text{F}_g\left(W_{\rm th}/\Omega_{\rm}\right)} \left. \frac{\partial P_{\rm op}\left(\Omega;\Gamma\right)}{\partial \Omega}\right|_{\Omega=\Omega_{\rm 0}}.
    \label{eqapp3}
\end{align}
\textcolor{black}{\noindent where $(a)$ and $(b)$ follow after applying the chain rule.}

From \eqref{eqop}, and applying again the chain rule in the integral definition of the OP, we have
\begin{equation}
    \left. \frac{\partial P_{\rm op}\left(\Omega;\Gamma\right)}{\partial \Omega}\right|_{\Omega=\Omega_{\rm 0}}=\left.-\frac{W_{\rm th}}{\Omega^2}f_g\left(W_{\rm th}/\Omega_{\rm}\right)\right|_{\Omega=\Omega_{\rm 0}}.\label{eqapp4}
\end{equation}
Finally, combining \eqref{eqapp3} and \eqref{eqapp4}, the proof is complete. 

\section{\textcolor{black}{ODO under SC reception}}
\textcolor{black}{Let $W_{\rm SC}=g_{\rm SC}\Omega$ be the instantaneous received power after SC in an $N$-branch receiver with i.i.d. branches \cite{Brennan1959}, with $\Omega$ representing the average received power per branch. Noting that $F_{g_{\rm SC}}(x)=F_{g}(x)^N$ and $f_{g_{\rm SC}}(x)=Nf_{g}(x)F_{g}(x)^{N-1}$, the following identity holds:}
\begin{equation}
    \textcolor{black}{\delta_{\rm SC} = \frac{W_{\rm th}}{\Omega_{\rm 0}}\frac{ f_{g_{\rm SC}}\left(W_{\rm th}/\Omega_{\rm 0}\right)}{{\rm F}_{g_{\rm SC}}\left(W_{\rm th}/\Omega_{\rm 0}\right)} = N \delta.}
\end{equation}
\label{app2}

\bibliographystyle{ieeetr}
\bibliography{refs}

\begin{thebibliography}{10}

\bibitem{Dang2020}
S.~Dang, O.~Amin, B.~Shihada, and M.-S. Alouini, ``{What should 6G be?},'' {\em Nat. Electron.}, vol.~3, pp.~20--29, Jan. 2020.

\bibitem{Ventura1997}
J.~Ventura-Traveset, G.~Caire, E.~Biglieri, and G.~Taricco, ``{Impact of diversity reception on fading channels with coded modulation. I. Coherent detection},'' {\em IEEE Trans. Commun.}, vol.~45, pp.~563--572, May 1997.

\bibitem{Wang2003}
Z.~Wang and G.~Giannakis, ``{A simple and general parameterization quantifying performance in fading channels},'' {\em IEEE Trans. Commun.}, vol.~51, pp.~1389--1398, Aug. 2003.

\bibitem{New2024}
W.~K. New, K.-K. Wong, H.~Xu, K.-F. Tong, and C.-B. Chae, ``{Fluid Antenna System: New Insights on Outage Probability and Diversity Gain},'' {\em IEEE Trans. Wireless Commun.}, vol.~23, pp.~128--140, Jan. 2024.

\bibitem{Dohler2011}
M.~Dohler, R.~Heath, A.~Lozano, C.~Papadias, and R.~Valenzuela, ``{Is the PHY layer dead?},'' {\em IEEE Commun. Mag.}, vol.~49, pp.~159--165, Apr. 2011.

\bibitem{Eggers2019}
P.~C.~F. Eggers, M.~Angjelichinoski, and P.~Popovski, ``{Wireless Channel Modeling Perspectives for Ultra-Reliable Communications},'' {\em IEEE Trans. Wireless Commun.}, vol.~18, pp.~2229--2243, Mar. 2019.

\bibitem{Ramirez2020}
P.~Ramírez-Espinosa, R.~J. Sánchez-Alarcón, and F.~J. López-Martínez, ``{On the Beneficial Role of a Finite Number of Scatterers for Wireless Physical Layer Security},'' {\em IEEE Access}, vol.~8, pp.~105055--105064, June 2020.

\bibitem{Sanayei2007}
S.~Sanayei, A.~Hedayat, and A.~Nosratinia, ``{Space Time Codes in Keyhole Channels: Analysis and Design},'' {\em IEEE Trans. Wireless Commun.}, vol.~6, pp.~2006--2011, Jan. 2007.

\bibitem{Erceg1997}
V.~Erceg, S.~Fortune, J.~Ling, A.~Rustako, and R.~Valenzuela, ``{Comparisons of a computer-based propagation prediction tool with experimental data collected in urban microcellular environments},'' {\em IEEE J. Sel. Areas Commun.}, vol.~15, pp.~677--684, May 1997.

\bibitem{Hasna2003}
M.~Hasna and M.-S. Alouini, ``{End-to-end performance of transmission systems with relays over Rayleigh-fading channels},'' {\em IEEE Trans. Wireless Commun.}, vol.~2, pp.~1126--1131, Nov. 2003.

\bibitem{Griffin2008}
J.~D. Griffin and G.~D. Durgin, ``{Gains For RF Tags Using Multiple Antennas},'' {\em IEEE Trans. Antennas Propag.}, vol.~56, pp.~563--570, Feb. 2008.

\bibitem{Safari2008}
M.~Safari and M.~Uysal, ``{Cooperative diversity over log-normal fading channels: performance analysis and optimization},'' {\em IEEE Trans. Wireless Commun.}, vol.~7, pp.~1963--1972, May 2008.

\bibitem{Elamassie2020}
M.~Elamassie and M.~Uysal, ``{Incremental Diversity Order for Characterization of FSO Communication Systems Over Lognormal Fading Channels},'' {\em IEEE Commun. Lett.}, vol.~24, pp.~825--829, Feb. 2020.

\bibitem{Narasimhan2006}
R.~Narasimhan, ``{Finite-SNR Diversity–Multiplexing Tradeoff for Correlated Rayleigh and Rician MIMO Channels},'' {\em IEEE Trans. Inf. Theory}, vol.~52, pp.~3965--3979, Aug. 2006.

\bibitem{Durgin2002}
G.~Durgin, T.~Rappaport, and D.~de~Wolf, ``{New analytical models and probability density functions for fading in wireless communications},'' {\em IEEE Trans. Commun.}, vol.~50, pp.~1005--1015, June 2002.

\bibitem{Brennan1959}
D.~G. Brennan, ``{Linear Diversity Combining Techniques},'' {\em Proc. IRE}, vol.~47, no.~6, pp.~1075--1102, 1959.

\bibitem{Romero2008}
J.~M. Romero-Jerez and A.~J. Goldsmith, ``{Receive Antenna Array Strategies in Fading and Interference: An Outage Probability Comparison},'' {\em IEEE Trans. Wireless Commun.}, vol.~7, no.~3, pp.~920--932, 2008.

\bibitem{Rao2015}
M.~Rao, F.~J. Lopez-Martinez, M.-S. Alouini, and A.~Goldsmith, ``{MGF} approach to the analysis of generalized two-ray fading models,'' {\em IEEE Trans. Wireless Commun.}, vol.~14, pp.~2548--2561, May 2015.

\bibitem{Matolak2011}
D.~W. Matolak and J.~Frolik, ``Worse-than-{R}ayleigh fading: Experimental results and theoretical models,'' {\em IEEE Commun. Mag.}, vol.~49, no.~4, pp.~140--146, 2011.

\bibitem{Gesbert2002}
D.~Gesbert, H.~Bolcskei, D.~Gore, and A.~Paulraj, ``{Outdoor MIMO wireless channels: models and performance prediction},'' {\em IEEE Trans. Commun.}, vol.~50, no.~12, pp.~1926--1934, 2002.

\end{thebibliography}

\vfill

\end{document}